\documentclass[format=acmsmall, review=false]{acmart}
\usepackage{acm-ec-21}
\usepackage{cleveref}
\usepackage{booktabs} 
\usepackage[ruled]{algorithm2e} 

\SetAlFnt{\small}
\SetAlCapFnt{\small}
\SetAlCapNameFnt{\small}
\SetAlCapHSkip{0pt}
\IncMargin{-\parindent}

\setcitestyle{acmnumeric}

\usepackage[skins]{tcolorbox}
 
\newtheorem{theorem}{Theorem}

\newtheorem{lemma}{Lemma}

\newtheorem{claim}{Claim}

\newtheorem{proposition}{Proposition}
\newtheorem*{proposition*}{Proposition}

\usepackage{enumitem}

\newcommand{\EE}{\mathbb E}
\newcommand{\PP}{\mathbb P}
\newcommand{\Aug}{\textnormal{\textsc{Aug}}}
\newcommand{\Adj}{\textnormal{\textsc{Adj}}}

\newcommand{\Occ}{\textnormal{\textsc{Occ}}}

\newcommand{\Rem}{\textnormal{\textsc{Rem}}}

\newcommand{\stablematching}{\textsc{SM}}
\newcommand{\greedycommit}{\textsc{Greedy-Commit}}
\newcommand{\opt}{OPT}
\newcommand{\OPT}{OPT}

\newcommand{\RoundAlgSucc}[1]{A_{\leq #1}}       
\usepackage{soul}

\usepackage{bbm}
\usepackage{nicefrac}
\usepackage{todonotes}
\newcounter{todocounter}

\title{Decentralized Matching in a Probabilistic Environment}

\author{Mobin  Y. Jeloudar, Irene Lo, Tristan Pollner, Amin Saberi$^*$}
\begin{document}

\begin{abstract}

 We consider a model for repeated stochastic matching where compatibility is probabilistic, is realized the first time agents are matched, and persists in the future. Such a model has applications in the gig economy, kidney exchange, and mentorship matching.  
  
We ask whether a \textit{decentralized} matching process  can approximate the optimal online algorithm. In particular, we consider a decentralized \textit{stable matching} process where agents match with the most compatible partner who does not prefer matching with someone else, 
and known compatible pairs continue matching in all future rounds. 
We  demonstrate that  the above process provides a 0.316-approximation to the optimal online algorithm for matching on general graphs. We also provide a $\nicefrac{1}{7}$-approximation for many-to-one bipartite matching, a $\nicefrac{1}{11}$-approximation for capacitated matching on general graphs, and a $\nicefrac{1}{2k}$-approximation for forming teams of up to $k$ agents. Our results rely on a novel coupling argument that decomposes the successful edges of the optimal online algorithm in terms of their round-by-round comparison with stable matching.
\end{abstract}

\begin{titlepage}

\maketitle

\end{titlepage}

\section{Introduction}

We consider a model with a finite set of agents who can be repeatedly matched in each of a finite number of rounds. Each pair of agents $(i,j)$ is compatible with a known probability $p_{\{i,j\}}$. When a pair is first matched, their compatibility is realized, and is successful with probability $p_{\{i,j\}}$ and unsuccessful with probability $1-p_{\{i,j\}}$. This compatibility persists through all future rounds.
This model captures learning dynamics in  platforms  that match workers with repeated tasks (dog-walking, babysitting, private chefs), mentorship programs, and kidney exchange. In such models, agents may only have an estimate of their compatibility with potential partners and typically learn their compatibility with others through being matched.



The goal of a matching platform is to maximize a \emph{weighted sum of the size of the matching in each round}. For example, a platform may want to maximize the total number of successful matches, such as the total number of dog-dogwalker days or babysitting events. Or, it may only be interested in the number of successfully matched pairs by the end of the matching process. In mentorship and kidney exchange programs agents delays are costly, so matches made earlier in the process may be more valuable than matches made later.
In contrast, the goal of each agent is to selfishly maximize the number of rounds in which they are matched.

Consider two matching processes. The optimal centralized matching algorithm, \opt,  solves an NP-hard Bayesian optimization problem to maximize the expected reward, by  prescribing matchings adaptively across rounds.   In contrast, in the Stable Matching or \stablematching\ process, the pairs that were successfully matched to each other in previous rounds remain matched and the rest of the matching is formed by pairing the agents with the highest success probability in a greedy fashion. Ties are broken arbitrarily. We show that \stablematching\ provides a constant-factor approximation of \opt.

\begin{theorem}\label{thrm:2.31}
The expected reward of \stablematching\ is at least a 0.316-approximation of the expected reward of the optimal online algorithm.
\end{theorem}




In \Cref{sec:model}, we provide a more detailed explanation of the \stablematching\ process and its relationship with stable matching. Informally, if agents form preference lists over all other agents based on probability of compatibility, then \stablematching\ forms a stable matching with respect to these preferences.

Stable matching has a number of attractive properties. While the optimal centralized matching \opt\ is  NP-hard to compute, stable matching can be computed in polynomial time and reached in a decentralized manner. Additionally, also in contrast with \opt, stable matching is incentive compatible: it does not match agents against their will by forcing unmatched agents with high compatibility to match with low-compatibility partners, or break up compatible matches in future rounds. Stable matching is also attractive in settings such as kidney exchange or mentor matching where time is valuable, as it does not sacrifice present potential matches for future payoffs. 


\smallskip

Our analysis implies a better approximation factor for another previously-known matching algorithm, known as \greedycommit \cite{chen2009approximating}. Similar to \stablematching,  \greedycommit\ commits to known compatible pairs by matching them in all future rounds, but unlike SM, it selects a maximum weight matching between the remaining agents to maximize the immediate expected reward. 
We can interpret \greedycommit\ as an off-the-shelf centralized matching service that can be computed by the platform in polynomial time if it knows the compatibility probabilities but does not want to solve the optimal online repeated matching problem. 
We show that \greedycommit\  provides a better constant-factor approximation to \opt.
\begin{theorem}\label{thrm:greedy-commit}
The expected reward of \greedycommit\ is at least a 0.43-approximation of the expected reward of the optimal online algorithm.
\end{theorem}

Finally, we consider settings where agents have the capacity to be bilaterally matched with multiple other agents, or grouped in teams. Many gig economy markets and matching settings are built on many-to-one bipartite markets, where agents on one side (e.g. workers) can be matched to multiple agents on the other side of the bipartition (e.g. tasks). In this setting \stablematching\ provides a $\nicefrac{1}{7}$-approximation to OPT. In the general matching problem where all agents can be matched to multiple other agents, \stablematching\ provides a $\nicefrac{1}{11}$-approximation to OPT, and in the teams problem where we form teams of up to $k$ agents, \stablematching\ provides a $\nicefrac{1}{2k}$-approximation to OPT. 
 
Together, our results contribute to the discussion of the benefits of investing in optimal centralized matching. We show if the platform allows myopic, self-interested agents to determine a decentralized matching, the outcome attains a constant factor of the expected reward of the optimal online matching. If the platform offers a centralized matching service, in addition to determining match probabilities it must also force agents into matches that are not incentive-compatible, and may even need to solve an NP-hard problem.

\subsection{Technical Ingredients}
\smallskip

A key technical contribution of this paper is to compare the decentralized stable matching process to the optimal \emph{online} algorithm for centralized repeated matching. In the study of online Bayesian optimization, an extremely common benchmark is the optimal \emph{offline} algorithm, which knows at the outset which edges are successful and unsuccessful.
In our setting, optimal offline is not an interesting benchmark, because for a certain family of inputs any online algorithm achieves an arbitrarily poor approximation to the optimal offline algorithm 
(see Appendix A). 



We  compare the stable matching \stablematching\ with the optimal online algorithm OPT by \emph{coupling} edges selected by OPT with those selected by \stablematching. 
Since the two algorithms uncover \textit{different information} as they progress, care must be taken in the coupling process not to condition analysis of the reward achieved by one algorithm on information that is acquired by the other.

\paragraph{Domination Lemma. }
A pivotal piece of our analysis is a \emph{Domination Lemma}, which uses the greedy structure of the stable matching to bound the reward generated by a subset of edges selected by OPT by twice the reward generated by \stablematching. The subset of edges is carefully chosen so that \stablematching\ approximately improves upon OPT, despite the fact that OPT is able to use information from prior rounds in a more sophisticated manner. 
Specifically, note  that \stablematching\ greedily selects max-weight edges between agents who have not yet been matched. Now fix a round $t$, and consider the edges selected by OPT in round $t$ that do not share an endpoint with successful edges previously selected (i.e. selected in rounds $1,2,\ldots,t-1$) by \stablematching. Such edges can also be selected by \stablematching\ in round $t$, and so the greediness of \stablematching\ implies that the expected number of such edges that are successful is bounded by the twice the expected number of successful edges newly selected by \stablematching\ in round $t$. While the intuition behind the Domination Lemma is straightforward, care needs to be taken to ensure that expectations are taken appropriately, since OPT and \stablematching\ explore different parts of the sample space and so their performance is conditional on different histories.  
\paragraph{Charging Lemma. } We bound the remaining edges using a \emph{Charging Lemma} that charges edges selected by OPT to adjacent edges selected by \stablematching. 
For matching on general graphs, the bounds provided by the Charging Lemma and a refined version of the Domination Lemma define a factor-revealing LP that yields the 0.316-approximation. In the other settings, generalized versions of the Charging Lemma and Domination Lemma show the constant-factor approximations.

\smallskip

Our analysis also provides some intuition for why stable matching, which utilizes a myopic, decentralized matching process in each round, and which limits its use of information by committing to successful matches, is nonetheless able to achieve a constant-factor approximation to the optimal online algorithm, which can select optimal matchings in each round, and can also adaptively make use of information across rounds. In the uncapacitated matching setting, the Domination Lemma shows that in any given round, agents who have not previously been matched can match themselves via stable matching at least half as well as \opt, and the Charging Lemma observes that in any given round, agents who are matched in that round via stable matching are collectively matched at least half as well compared to \opt. Therefore stable matching cannot do much worse than \opt, even though it makes use of information in a much less sophisticated manner.

\subsection{Relationship with Prior Work}\label{subsec:related-literature}

There is a large body of work on \emph{centralized} and \emph{decentralized} matchings in online and two-sided platforms. Most of this literature focuses on designing optimal information structures in centralized matchings, or optimal procedures for centralized matching by the platform. Our question also motivates a variant of \emph{repeated stochastic matching} that is related to existing literature on the \emph{query-commit problem} and \emph{stochastic matching with rewards}. Interestingly, our analysis is sufficiently general to also handle previously studied centralized algorithms in the query-commit setting, and consequently we are able to provide improved bounds on prior algorithms.

\vspace{-5pt}

\paragraph{Optimal matching in matching platforms} There is also a substantial literature on finding optimal or near-optimal matching policies in two-sided platforms with long-lived agents, such as matching in ridesharing \cite{gurvich2015dynamic, banerjee2015pricing,  banerjee2016pricing,hu2020dynamic},  volunteer platforms \cite{manshadi2020online}  and blood donation \cite{mcelfresh2020matching}. We focus on such platforms where match compatibility is difficult to determine and can be learned exactly only through matching. A well-known example of such a setting is kidney exchange, which has been studied from a repeated matching perspective \cite{ashlagi2019matching, akbarpour2020thickness} and a failure-aware perspective \cite{dickerson2013failure}. 
\vspace{-5pt}
\paragraph{Online and Stochastic Matching.} The online bipartite matching problem introduced in \cite{karp1990optimal}, where vertices on one side of a bipartite graph arrive online, is foundational to the literature on online matching problems. Many variations have been studied, including the \emph{adwords problem} \cite{mehta2007adwords, goel2008online}, \emph{matching with stochastic rewards}  \cite{mehta2012online,mehta2014online,goyal2019online,huang2020online} (where edge realizations are stochastic) and the random-order and i.i.d. online matching problems  \cite{devanur2009adwords, feldman2009online, korula2009algorithms, karande2011online, mahdian2011online, manshadi2012online,kesselheim2013optimal,jaillet2014online}. While almost all of these papers provide competitive guarantees compared to the optimum offline algorithm, we  provide guarantees compared to the optimum online algorithm in a multi-round environment. 

\vspace{-5pt}
\paragraph{Query-Commit.} Motivated in part by the application of kidney exchange, there is a large body of literature studying a variant of stochastic matching known as the query-commit problem \cite{chen2009approximating, adamczyk2011improved, bansal2012lp, molinaro2011query, goel2012matching,gupta2017adaptivity,gamlath2019beating}. Here, edges have fixed realization probabilities and can be queried in sequence, with the constraint that successful edges queried must be accepted, and accepted edges must form a matching. The objective is to maximize the size of the resulting matching. \cite{chen2009approximating} introduced this setting and proved that the greedy algorithm (which simply queries feasible edges in decreasing order of probability) is a $\nicefrac{1}{4}$-approximation to the optimal online algorithm even when vertices have \emph{patience parameters}, i.e. limits on how many incident edges can be queried. The authors additionally showed that the optimal online algorithm is NP-hard to compute. \cite{adamczyk2011improved} improved on this result by showing that the greedy algorithm is in fact a $\nicefrac{1}{2}$-approximation (with patience parameters), and \cite{bansal2012lp} provided constant-factor approximations on weighted graphs through an LP-based approach. Later work provided constant-factor approximations in query-commit settings for feasibility constraints significantly generalizing the matching constraint (e.g., \cite{gupta2013stochastic}, \cite{gupta2016algorithms}, \cite{adamczyk2016submodular}).

A less common assumption in the query-commit literature is the ability to query a matching in each round, instead of a single edge. \cite{chen2009approximating} define this setting, propose the \greedycommit\ algorithm, and prove that it achieves a $\nicefrac{1}{4}$-approximation to the optimal online algorithm which is forced to commit. Their analysis can furthermore be extended to a \emph{capacitated matching} setting, where in each round we can select a matching that must respect integral capacity constraints on each vertex. \cite{bansal2012lp} also prove a $\nicefrac{1}{20}$-approximation in the weighted setting, where the constraint is that a matching of size at most $C$ can be queried in each round (for any parameter $C$). 

The main differences between this line of work and our contribution are two-fold. Most importantly, while \greedycommit\ operates \emph{globally} in each round, the decentralized matching process we analyze cannot replicate such a centralized greedy approach. 
Secondly, we compare to algorithms that are not forced to commit to successful edges.
(If OPT-Commit denotes the optimal online algorithm that is forced to forever match any successful edges it finds, simple examples show that OPT is a more powerful algorithm than OPT-Commit; see Appendix C
.) Because we compare to the optimal online algorithm that is not restricted to committing, our coupling argument is substantively different to, and more delicate than, the one in \cite{chen2009approximating}. Despite this challenge, our analysis of the decentralized setting still improves upon \cite{chen2009approximating}; we show greedy achieves a 0.316-approximation against a non-committing (more powerful) OPT, tightening the previous analysis of a $\nicefrac{1}{4}$-approximation to committing OPT.

\vspace{-5pt}
\paragraph{Learning by Matching.} Prior work has also explored several models of learning in matching settings. Many of these papers consider a multi-armed bandit setting where rewards are stochastic and redrawn i.i.d. from an unknown distribution, with algorithms of an explore-exploit nature \cite{Channel_Allocations,johari2016matching,Bandits_Matching_Markets}.  Efficient algorithms that approximate the optimal online algorithm were also studied by \cite{price_of_information} in a stochastic setting with a price for querying each edge, relying on a Gittins Index characterization for the optimal online algorithm \cite{gittins1974dynamic, weitzman1979optimal}. 

\section{Stable Matching in General Graphs}\label{sec:model} 
We begin with a setting where each agent is matched with at most one other agent in each round. This setting captures bipartite matching problems such as matching mentors and mentees, as well as matching in general graphs such as matching peer mentors, roommates and kidney exchange.

There is a set $V$ of vertices, representing agents. For every pair $i, j \in V$ of agents with $i \neq j$, there is an edge between $i$ and $j$ with probability $p_{\{i,j\}}$ independent from other edges, representing the compatibility of the agents. The set of agents and probabilities are known at the outset. We will find it useful to view the entire graph as being generated randomly from the outset: before the first round nature samples the graph $G(V,E)$ with probability $\prod_{e \in E} p_e \prod_{e \notin  E} (1-p_e).$ The platform and agents do not have direct access to $G$. Instead, in rounds $1 \leq t \leq T$, they can determine a matching $M_t$ between the vertices in $V$ and observe whether the edges in $M_t$ are in $E$ or not.

Let $X_e$ be the Bernoulli random variable indicating whether $e$ is in $E$; we say an edge is {\em successful} if $X_e = 1$.
Each agent's goal is to maximize the expected number of rounds in which they are successfully matched. 
The platform's reward in each round is equal to the number of successful edges selected in that round. 
Given weights $\omega_1,\ldots,\omega_T$, the platform's goal is to maximize the weighted sum $\sum_{t=1}^T \omega_t \sum_{e\in M_t} X_e$ of the rewards collected in all $T$ rounds. This can capture if the platform's goal is to maximize the size of the matching in the last round ($\omega_t = \mathbbm{1}_{t=T}$), maximize the total number of successful matches in each round ($\omega_t = 1$ for all $t$), or weighted toward favoring matches in earlier rounds (e.g. $\omega_t=\delta^t$ for $\delta \in (0,1)$).

The (deterministic) optimal online algorithm for maximizing the total reward can be derived by an exponential-size dynamic program \cite{bertsekas1995dynamic}. However, exact optimization is NP-hard (as we discuss later), and for large numbers of agents it is hence infeasible for a matching platform to compute this algorithm. With this in mind, we focus on simple  matching processes that are computable in polynomial time, and give approximations to the optimal online algorithm.

We assume
that each agent $i$ knows his or her compatibility probabilities $p_{\{i,j\}}$ with every other agent $j$. 
This could be enabled, for example, by search functionality in the platform that allows agents to view information about other agents.
 Fix a given round and an agent $i$, and assume that through prior matches agent $i$ now has \emph{updated} their priors with other agents to $\{\hat{p}_{\{i,j\}}\}$ (i.e., all unsuccessful matches are updated to 0 and all successful matches are updated to 1). Agent $i$ wants to be matched as soon as possible, so in this round would like to match with the remaining agent $\text{argmax}_{j\neq i} \hat{p}_{\{i,j\}}$ who is most likely to be compatible with $i$. Hence, agent $i$ can form a preference list over all other agents in this round by sorting their current compatibility probabilities $\{\hat{p}_{\{i,j\}}\}$.
 
 In the \stablematching\ process, we assume that in each round agents choose a matching with no \emph{blocking pair} under these preferences (two agents who are not matched to each other but  prefer each other to their matched partners). We additionally assume that once a pair of agents is successfully matched they will continue matching with each other in future rounds. It is straightforward to see that the \emph{decentralized stable matching}, where there are no blocking pairs, can be formalized as follows.

 \begin{tcolorbox}
 
 \textbf{Stable Matching (\stablematching)}
 
Initialize the set of successful edges $A \leftarrow \emptyset$. Let $S \leftarrow V.$ 

For rounds $1 \leq t \leq T$:  

\begin{itemize} 
    \item While there are agents in $S$ who are compatible with positive probability, find the most compatible pair $f$, match $i$ with $j$, and remove $i,j$ from $S$. Let $M_t$ be the matching determined once no agents in $S$ are compatible with positive probability.
    \item Output $M_t \cup A$  as the selected matching for round $t$.
    \item If an edge $e$ in $M_t$ is successful add it to $A$. Otherwise add both its endpoints back to $S$ and set $p_e \leftarrow 0$.
\end{itemize}
\end{tcolorbox}

In other words, the stable matching can be determined by greedily selecting pairs of agents who are most likely to be compatible, matching them, and committing to matching them in future rounds if their edge is successful.
The stable matching is incentive-compatible \emph{within-rounds}, in the sense that if in round $t$ some agent $i$ prefers matching with another agent $j$ to their match in $M_t$ then $j$ is matched to a preferred agent. The stable matching is also incentive-compatible \emph{across-rounds}, in the sense that agents who are successfully matched in previous rounds prefer staying matched to matching again, and agents always prefer to be matched as soon as possible.

For additional clarity, consider when all compatibility probabilities are distinct and strictly smaller than 1. In this case, the stable matching in the first round is unique: the pair with the highest compatibility  match with each other (otherwise they would form a blocking pair), the pair with the highest probability among the remaining agents match, and so on. No matter the realizations, in all rounds there is a unique stable matching based on the updated preferences.
When some edge probabilities are the same the stable matching may not be unique, but our analysis applies no matter which stable matching is selected, as long as previously matched edges stay matched.

We compare the result of \stablematching\ with a platform that optimizes its matching process centrally. One downside to this approach is that determining this \emph{optimal online algorithm (\opt)} is NP-hard.
\begin{proposition}\label{lemma:np-hard}
Computing the optimal online algorithm is NP-hard.
\end{proposition}

This was originally proved by Chen et al. \cite{chen2009approximating};\footnote{\label{fn:chen}\cite{chen2009approximating} state the hardness result for a slightly different setting, where the optimal online algorithm is forced to commit to edges that are successful. Despite the fact that in our setting, \opt\ is not required to commit, the same proof holds.} for completeness we provide a full proof for our setting in Appendix D.
This hardness result motivates the analysis of matching algorithms which can be computed efficiently. Our main result is the following.

\setcounter{theorem}{0}
\begin{theorem}
\label{thrm:2.31} 
The expected reward of stable matching is at least a 0.316-approximation of the expected reward of the optimal online algorithm.
\end{theorem}
\setcounter{theorem}{2}

The proof of Theorem \ref{thrm:2.31} relies on coupling edges selected by \OPT\ with those selected by \stablematching. A subset of these edges are bounded using a \emph{Domination Lemma}, which uses the greediness of \stablematching\ to bound the expected reward generated by a subset of edges selected by \OPT\ by twice the expected reward generated by \stablematching. This coupling requires a delicate comparison of two algorithms that may have, through different histories, discovered very different pieces of information about the sample graph $G$. The remaining edges are bounded using a \emph{Charging Lemma} that charges edges selected by \OPT\ to adjacent edges selected by \stablematching. 
The bounds provided by the Charging Lemma and the Domination Lemma define a factor-revealing LP that yields the 0.316-approximation. 
In \Cref{sec:upperbound}, we additionally show an upper bound of $0.5+\epsilon$ on the approximation ratio that \stablematching\ achieves to \OPT.

\subsection{Coupling Edges Selected by Stable Matching  and OPT} \label{sec:notation}

The main technical innovation in our paper is to couple the edges selected by OPT with the edges selected by \stablematching\ in a way that admits a constant-factor bound in expectation for each round \textit{despite the fact that \stablematching\ and OPT learn different information in each round}. To specify the subsets of edges that are coupled in the Charging and Domination Lemmas, we introduce some additional notation.

Throughout the paper, for a random set $S$ we will write $\EE[S]$ as a shorthand for $\EE[|S|]$. We use $\sqcup$ to denote a union of sets that are disjoint. 
 
For all $t \in [T]$, let $S_{\le t}$ denote the set of \textit{successful} edges selected by \stablematching\ in round $t$.\footnote{Note $S_{\le t}$ is a random set that is a deterministic function of the sample graph $G$.}
Also define $S_t$ to be the successful \emph{new} edges that \stablematching\ selects for the first time in round $t$. We hence have
\begin{equation}
S_{\le t} = S_1 \sqcup S_2 \sqcup \ldots \sqcup S_t.
\end{equation} Similarly, let $O_{\leq t}$ denote the set of \textit{successful} edges that \OPT\ selects in round $t$. 

Fix $t$. To bound the expected number of edges in $O_{\le t}$, we will partition these edges based on whether or not they can be added to $S_{\le t}$. In particular, we write
\begin{equation}
\label{eq:aug1}
O_{\le t} = \textsc{Aug} \sqcup \textsc{Adj}
\end{equation}
where $\textsc{Aug}$ denotes all the edges in $O_{\le t}$ that are vertex-disjoint from $S_{\le t}$ (and hence can augment this matching), and $\Adj$ denotes the remaining edges that share an endpoint with some edge in $S_{\le t}$.\footnote{As $S_{\le t}$ and $O_{\le t}$ are deterministic functions of the sample graph $G$, $\Aug$ and $\Adj$ are hence also random sets entirely determined by the sample graph.} It will also be useful to categorize all edges in $O_{\le t}$ according to the round in which OPT first selected them; hence, we will write
\begin{equation}
    \label{eq:aug2}
\Aug = \Aug_1 \sqcup \Aug_2 \sqcup \ldots \sqcup \Aug_t,
\end{equation}
where $\Aug_j$ denotes all edges in $\Aug$ that OPT first selected in round $j$. Similarly, we break up 
\begin{equation} 
\Adj = \Adj_1 \sqcup \Adj_2  \sqcup \ldots \sqcup \Adj_t
\end{equation} 
where $\Adj_j$ is all edges in $\Adj$ that OPT first selected in round $j$. 
Note that  $\Aug$ and $\Adj$ and their corresponding partitions are defined with respect to $S_{\le t}$ and $O_{\le t}$ and therefore they all depend on $t$. We have dropped the index $t$ from their notation to simplify the exposition.

In the rest of this section, we will use the Domination Lemma to couple edges in $S_i$  with edges in $\Aug_i$, and the Charging Lemma to couple edges in $S_t$ with edges in $\Adj_t$. This coupling provides bounds that define the factor-revealing LP.

\subsection{The Domination Lemma}

Our main technical lemma is the Domination Lemma, which makes use of the greediness of \stablematching\ to bound its approximation to OPT. In particular, recall that $S_{\le i-1}$ is the set of \textit{successful} edges selected by \stablematching\ in round $i$. Then in each round $i \leq t$, \stablematching\ selects edges to add on to $S_{\le i-1}$ greedily. In this step, the Domination Lemma lower bounds its expected reward versus some of the edges selected by \opt.

\begin{lemma}[Domination Lemma]\label{lemma:dominationlemma} Fix a round $t$ and define $\Aug_i$ for $1 \leq i \leq t$ as in Equations (\ref{eq:aug1}) and (\ref{eq:aug2}). Then, for all $i \leq t$, $2 \cdot \EE[S_i] \ge \EE[\Aug_i]$.
\end{lemma}

\begin{proof}

The expectations in the lemma are over all sample graphs. In the proof, we require a consideration of what each algorithm knows up to round $i$, and through this split up the probability space the expectations are being taken over. To do so, we define the notion of a history induced on an algorithm, and analyze the expected size of $S_i$ and $\Aug_i$ conditioned on these histories.

Fix some $i \le t$. For a sample graph $G$ and an algorithm $\mathcal{A}$, we define $H^{\mathcal{A}}(G)$, the \emph{history} that $G$ induces on $\mathcal{A}$, as follows. The history $H^{\mathcal{A}}(G)$ consists of the set of matchings $\mathcal{A}$ selects in the first $i-1$ rounds, along with the observed outcomes $\{X_e\}$ of all the edges in these matchings, conditioned on the sample graph being $G$. We next partition the probability space over sample graphs based on the distinct histories they induce on \stablematching\ in the first $i-1$ rounds. In particular, for any possible history $h$, let $\mathcal{G}_{h}$ denote the set of all sample graphs $G$ such that $H^{\stablematching}(G) = h$. 
We use $\EE[ S_i \mid \mathcal{G}_{h} ]$ as a shorthand for the expected size of $S_i$, conditioned on the sample graph $G$ being in $\mathcal{G}_{h}$, and use $\PP [\mathcal{G}_{h}]$ for the probability that our sample graph $G$ is an element of $\mathcal{G}_{h}$. Note then that 
$$\EE[S_i] = \sum_{h} \EE[S_i \mid \mathcal{G}_{h}] \cdot \PP[\mathcal{G}_{h}], ~~\text{and}~~~\EE[\Aug_i] = \sum_{h} \EE[\Aug_i \mid \mathcal{G}_{h}] \cdot \PP[ \mathcal{G}_{h}],$$ where the sum is over all possible histories $h$. 
It hence suffices to show that $2 \cdot \EE[S_i \mid \mathcal{G}_{h}] \geq \EE[\Aug_i \mid \mathcal{G}_{h}]  $ for any history $h$. 

To do so, we analyze $\EE[\Aug_i \mid \mathcal{G}_{h}] $ by partitioning $\mathcal{G}_{h}$ based on the different histories these sample graphs induce on OPT. Specifically, for a fixed history $h'$, let $\mathcal{G}_{h, h'}$ denote all sample graphs $G$ such that $H^{\stablematching}(G) = h$ and $H^{\opt}(G) = h'$. Note 
 that $\EE[\Aug_i \mid \mathcal{G}_{h}] = \sum_{h'} \EE[\Aug_i \mid \mathcal{G}_{h, h'}] \cdot  \PP[\mathcal{G}_{h, h'}| \mathcal{G}_h],$ 
and therefore, to complete the proof of Lemma \ref{lemma:dominationlemma}, it suffices to show that $$2 \cdot \EE[S_i \mid \mathcal{G}_{h}] \ge \EE[ \Aug_i \mid \mathcal{G}_{h, h'}]$$ for any possible histories $h$, $h'$. The rest of the proof is devoted to this inequality.


Conditioned on our sample graph being in $\mathcal{G}_{h, h'}$, let $N$ be the \textit{new} edges selected by OPT in round $i$ that can augment the matching $S_{\leq i-1}$. As we mentioned earlier, we can assume that OPT is deterministic. So $N$ is uniquely determined by $h$ and $h'$. Also, by definition $N$ is disjoint from all edges in $h'$ and all successful edges in $h$. 

Let $N_0 \subseteq N$ denote those edges in $N$ that belong to neither $h$ nor $h'$. Because edges are independent, for each edge $e$ outside history $h$ or $h'$, the probability that a random sample graph in $\mathcal{G}_{h, h'}$ includes edge $e$ is exactly $p_e$. Therefore, the expected number of edges in $N$ that are successful is given by $\sum_{e \in N_0} p_e$  as all edges in $N \setminus N_0$ are guaranteed to be unsuccessful.



Furthermore, observe that the set of edges in $N_0$ can feasibly augment $S_{\leq i-1}$ as every edge in $N_0$ is vertex-disjoint from $S_{\le i-1}$, and the edges in $N_0$ form a matching disjoint from $h$. Note that \stablematching\ chooses edges greedily to augment $S_{\leq i-1}$; as the greedy algorithm gives a $\nicefrac{1}{2}$-approximation to maximum weight matching, we hence have 
$$ 2 \cdot \EE[S_i \mid \mathcal{G}_h] \ge \sum_{e \in N_0} p_e.$$ Additionally, note that conditioned on our sample graph being in $\mathcal{G}_{h, h'}$, $\textsc{Aug}_i$ is guaranteed to be a subset of the successful edges in $N_0$, as edges augmenting $S_{\le t}$ must also augment $S_{\le i-1}$. Therefore, 
$$\sum_{e \in N_0} p_e \ge \EE[\Aug_i \mid \mathcal{G}_{h, h'}].$$ We conclude 
$$2 \cdot \EE[S_i \mid \mathcal{G}_h] \ge \EE[\Aug_i \mid \mathcal{G}_{h, h'}]$$ as desired. 
 \end{proof} 
 
 To obtain a bound of 0.316 on the approximation factor achieved by \stablematching\, we provide a refinement of the Domination Lemma.
    To prove the Domination Lemma, roughly we showed that for any fixed round $i$, if we look at the successful edges that OPT selects for the first time in round $i$ which can augment $S_{\le i-1}$, the expected size of this set is no more than $2 \cdot \EE[S_i]$. As $\Aug_i$ fits this description, this shows that $\EE[\Aug_i] \le 2\cdot \EE[S_i]$.     However, this lemma is loose; $\Aug_i$ does not necessarily comprise \emph{all} the edges that OPT selects for the first time in round $i$ which can augment $S_{\le i-1}$. In particular, certain edges in $\Adj_i$ might fit this description as well --- they certainly cannot augment $S_{\le t}$ by definition, but for small values of $i$ many of them might be able to augment $S_{\le i-1}$. 



To refine the Domination Lemma, we categorize the edges in $\Adj$ according to which of $S_1$, $S_2$, $\ldots$, $S_t$ they are incident to; this is motivated by the goal of finding some edges in $\Adj_i$ that can augment $S_i$. In particular, we break up 
$$\textsc{Adj} = \textsc{Adj}(S_1) \sqcup \Adj(S_2) \sqcup \ldots \sqcup \Adj(S_t)$$ where $\Adj(S_j)$ denotes the subset of $\Adj$ consisting of edges incident with $S_j$ but not $S_{\le j-1}$. This is well-defined because the sets $S_1$, $S_2$, $\ldots$, $S_{t}$ are disjoint and form a matching. Similarly, for any fixed $j$ we break up
$$\Adj(S_j) = \Adj_1(S_j) \sqcup \Adj_2(S_j) \sqcup \ldots \sqcup \Adj_t(S_j)$$ where $\Adj_i(S_j)$ is all edges in $\Adj(S_j)$ that OPT first selected in round $i$. 

The informal statement of the refined Domination Lemma is that if we fix a round $i$, and look at all edges that OPT first discovered in round $i$ that can be added to $S_{\le j-1}$, the expected size of this set is at most twice the expected size of $S_j$. Our notation above lets us write this formally. 

\begin{lemma}[Refined Domination Lemma]\label{lemma:domination} 
Fix a round $t$. We have for all $i,j \le t$ that $$\EE[\Aug_i] + \EE[\Adj_i(S_j)] + \EE[\Adj_i(S_{j+1})] + \ldots + \EE[\Adj_i(S_t)] \le 2 \cdot \EE[S_j]. $$
\end{lemma}

The proof of the Refined Domination Lemma only requires small changes from the proof of the Domination Lemma, and can be found in Appendix E.1.

\subsection{Proof of Theorem \ref{thrm:2.31}}

We next state and prove the Charging Lemma, related to the fact that two maximal matchings in a graph always have sizes within a factor of 2. 



\begin{lemma}[Charging Lemma] \label{lemma:charging2}
Fix a round $t$. For all $j \le t$, 
$$\sum_{i=1}^t \EE[\Adj_i(S_j)] \le 2 \cdot \EE[S_j].$$
\end{lemma}
\begin{proof} Fix a sample graph. Consider any edge $e \in \cup_i \textsc{Adj}_i(S_j)$. By definition, it is incident to at least one edge in $S_j$; charge $e$ to one of the edges in $S_j$ it is incident to. Note that because $\cup_i \textsc{Adj}_i(S_j)$
forms a matching, every edge in $S_j$ is charged at most twice. Hence the result holds sample graph by sample graph, as well as in expectation over all sample graphs.\end{proof} 

The task remaining is to use the structure we have found to give an upper bound on $\frac{\EE[O_{\le t}]}{\EE[S_{\le t}]}$; we do so via a factor-revealing linear program, with constraints corresponding to the Refined Domination Lemma and the Charging Lemma.  

\vspace{1em}
\textbf{Factor-revealing LP (Primal)}
\setcounter{equation}{0}
\begin{align}
\nonumber \max\;\;\;\;\;\; \sum_{i = 1}^{t} X_{i} +  \sum_{i = 1}^{t}  \sum_{j = 1}^{t} & X_{i,j} &&\text{ }  \\
\text{such that \;\;\;\;\;\;\;}
X_{i} + \sum_{q = j}^{t} X_{i,q} &\leq 2Y_j \;\; &&\text{for all } 1 \le i, j \le t  \\
\sum_{i = 1}^{t} X_{i,j} &\leq 2 Y_j \;\; &&\text{for all } 1 \le j \le t \\
\sum_{j = 1}^{t} Y_{j} &\leq 1 \;\; &&\text{  }\\
 Y_j \geq 0, \; X_{i} \geq 0, \; & X_{i,j} \geq 0  \;\; &&\text{for all }  1 \le i,j \le t 
\end{align}
\vspace{1em}

In this linear program, the variables $\{X_i\}_{1 \le i \le t}$ correspond to $\{ \EE[\Aug_i] \}$, the variables $\{X_{i,j}\}_{1 \le i,j \le t}$ correspond to $\{ \EE[\Adj_i(S_j)] \}$, and the variables $\{Y_j\}_{1 \le j \le t}$ correspond to $\{ \EE[S_j] \}$, and the bounds from the lemmas are encoded as constraints. 

Given any instance of our matching problem, we can construct a feasible solution for this LP by setting: 
$$X_{i} = \frac{\EE[\Aug_i]}{\EE[S_{\le t}]}, \; X_{i,j} = \frac{\EE[\Adj_i(S_j)]}{ \EE[S_{\le t}]}, Y_j = \frac{\EE[S_j]}{\EE[S_{\le t}]}.$$
Indeed, note that when we set the variables in this way, (1) directly states the Refined Domination Lemma, and (2) directly states the Charging Lemma. Also, because $\sum_j |S_j| = |S_{\le t}|$ we have that (3) holds. For the suggested feasible solution, we can  note that the objective simplifies to: $$\sum_{i = 1}^{t} X_{i} +  \sum_{i = 1}^{t}  \sum_{j = 1}^{T} 
X_{i,j}  = \frac{\sum_{i} \EE[\Aug_i] + \sum_{i, j} \EE[\Adj_i(S_j)] }{\EE[S_{\le t}]} = \frac{\EE[O_{\le t}]}{\EE[S_{\le t}]}$$ Hence, the maximum objective obtained by our LP gives a lower bound on the worst-case competitive ratio that \stablematching\ achieves against OPT. 

To give a bound on the maximum value obtained by this LP, we take the dual, noting that by weak duality it suffices to analyze the value obtained by a specific feasible solution. The dual of our LP is given below. 
\vspace{1em}

\textbf{Factor-revealing LP (Dual)}
\setcounter{equation}{0}
\begin{align}
\nonumber \min\;\;\;\;\;\;  u \hfill  && \text{ } \\
\text{such that \;\;\;\;\;\;\;}
\sum_{q = 1}^{j} F_{i,q} + c_j &\geq 1 \;\; &&\text{for all } 1 \le i, j \le t \\
\sum_{i = 1}^{t} 2F_{i,j} + 2c_j &\leq u \;\; &&\text{for all } 1 \le j \le t \\
\sum_{j = 1}^{t} F_{i,j} &\geq 1 \;\; &&\text{for all } 1 \le i \le t \\
 u \geq 0, \; c_{j} \geq 0, \; & F_{i,j} \geq 0  \;\; &&\text{for all } 1 \le i, j \le t
\end{align}
\vspace{1em}

Consider a feasible solution for the dual LP given by 
$F_{i,j} = \frac{1}{1 + t(e - 1)} \cdot \left( \frac{t}{t-1} \right)^{j-1} \text{ for } 1 \le i, j \le t $, 
 $c_j =   1 - \frac{t^j - (t-1)^j}{t(t-1)^{j-1}e - (t-1)^{j}} 
\text{ for } 1 \le j \le t $ and
$ u= \max_{j \in [t] } \left \{  \sum_i 2F_{i,j}+2c_j  \right \} = 2 + \frac{2(t-1)^t}{t^t - (t-1)^{t}}$. \footnote{This is the optimal Dual solution. Note that if we change our objective, we might be able to get better approximation factors; e.g. if we change the objective to be the cumulative reward, which sums the successful selected edges over every round instead of finding a round-by-round approximation guarantee, we will be able to achieve a 2.96-approximation factor.} The objective value $u = u(t)$ is increasing in $t$ and 
 $$\lim_{t\rightarrow\infty} u(t) = 2 + \frac{2}{\mathrm{e}-1}.$$

Hence by weak duality, the objective value of our primal is at most $2 + \frac{2}{e-1}$. It follows that the expected reward of \stablematching\ in round $t$ is at least a $ \left( 2 + \frac{2}{e-1} \right)^{-1} \ge 0.316$ fraction of the expected reward of OPT in round $t$. This completes the proof of Theorem~\ref{thrm:2.31}.

\subsection{Discussion: How much better can a matching platform do?}\label{sec:upperbound}

So far, we have shown that a decentralized stable matching process achieves at least a $0.316$ proportion of the reward of the optimal matching service. In this section, we provide some additional discussion on what else is achievable by the platform.


First, suppose the platform could invest in an off-the-shelf matching process. 
In particular, we consider the \greedycommit\ algorithm that was previously proposed by \cite{chen2009approximating}. The \greedycommit\ algorithm proceeds by proposing in the first round a matching of vertices in $V$ that maximizes the sum of edge probabilities, and in subsequent rounds committing to keeping all successful edges from the previous rounds, and augmenting them with a matching that maximizes the sum of edge probabilities in the remaining graph. We describe \greedycommit\ formally in Appendix B.

 
 


Note that \greedycommit\ is computable in polynomial-time, and is incentive-compatible across rounds but not within rounds. 
Hence \greedycommit\ can be thought of as a matching service that \emph{hides information} and myopically dictates a maximum matching in each round, but cannot incentivize matched agents to be rematched in future rounds.

Chen et al. showed that \greedycommit\ gives a $0.25$-approximation to \opt.\footnote{
\cite{chen2009approximating} provide a proof for a slightly different result which works to prove Lemma~\ref{lem:UB}; see footnote~\ref{fn:chen}.
}  As an extension to our other results in this section, we show that our techniques improve on the analysis of Chen et al. In particular, we show the following result.

\setcounter{theorem}{1}
\begin{theorem}\label{thrm:greedy-commit}
The expected reward of \greedycommit\ is at least a 0.43-approximation to \opt.
\end{theorem}
\setcounter{theorem}{3}


The proof proceeds very similarly to the proof of the approximation factor for \stablematching. The Charging Lemma continues to hold, and a stronger version of the Refined Domination Lemma holds, improved by a factor of 2. These constraints yield a different factor-revealing LP which we use to prove the claimed 0.43 factor. Details can be found in Appendix E.2.



Another natural question is what upper bound we can give on how well \stablematching\ and \greedycommit\ can perform compared to the optimum online algorithm. Although computing \opt\ is NP-hard, perhaps one of them could be very close to optimal, in the sense of guaranteeing a $(1-\varepsilon)$-approximation to \opt\ for some small $\varepsilon$? We show this is not the case; in particular, we show that both \stablematching\ and \opt\ provide at best a $\nicefrac{1}{2}$-approximation to the optimal online algorithm.

\begin{lemma}\label{lem:UB}
The expected reward of \stablematching\ and \greedycommit\ are at most a $\nicefrac{1}{2}$-approximation to \opt.
\end{lemma}

\begin{proof}
Consider the bipartite graph $G_{n, \epsilon}$ with bipartition $(U=(u_1,\ldots,u_n),V=(v_1,\ldots,v_n))$ edges $E=\{u_1v_i\}\cup\{u_iv_1\}$, the edge $u_1 v_1$ has probability 1 of being successful, and every other edge has probability $0.5 - \epsilon$. We illustrate the graph in Figure \ref{fig:Greedy2approx} in the appendix. 

Suppose there are $T = n^2$ rounds. 
\stablematching\ matches $u_1$ to $v_1$ in the first round. Since the edge is successful, $u_1$ will be matched to $v_1$ in every subsequent round. Hence \stablematching\ obtains a reward of $n^2$. 

We now define a different algorithm that selects the matching $(u_1 v_{i+1}) \cup (u_{i+1} v_1)$ in round $i$, for $1 \le i \le n-1$. With probability $1 - o(1)$, both edges are successful in at least one of these rounds. In this case the algorithm can select this matching from round $n$ to round $T$, obtaining payoff $2n^2(1 - o(1))$ in these rounds. In expectation, the reward of the above algorithm is at least $(1-o(1)) \cdot 2n^2(1-o(1)) = 2n^2(1 - o(1)).$ As \stablematching\ only obtains a reward of $n^2$, taking $n$ sufficiently large proves the claim. 

The same example proves that \greedycommit\ is at most a $(\nicefrac{1}{2} + \epsilon)$-approximation to OPT. 
\end{proof}

Finally, one might wonder how much of this gap between stable matching or \greedycommit\ and \opt\ is due to the fact that both algorithms commit to previously matched successful pairs, while \opt\ does not.
Consider the optimal online algorithm that is restricted to committing to successfully matched pairs \emph{(\opt-Commit)}.
\opt-Commit can be thought of as the optimal matching service that improves on \greedycommit\ by dictating a matching in each round in an optimal \emph{non-myopic} manner, but again cannot incentivize matched agents to be rematched in future rounds.
We note that the expected reward achieved by \OPT-Commit is not always the same as \opt\      (see Appendix C
), so committing causes some loss of reward.

Furthermore, we show 
that  OPT-Commit achieves at least a 0.5-approximation to \opt. 
Since \stablematching\ and \greedycommit\ achieve at best a $0.5$-approximation to \opt, if they do not achieve this then part of the loss is due to the fact that both \stablematching\ and \greedycommit\ are \emph{myopic} and do not account for the effect of current queries on possible matchings in future rounds, rather than the fact that both commit to previously matched successful pairs.
\begin{proposition}\label{prop:opt-optcommit}
The expected reward of \opt-commit is at least a $\nicefrac{1}{2}$-approximation to \opt.
\end{proposition}
The proof of Proposition~\ref{prop:opt-optcommit} relies on constructing an algorithm that commits to successfully matched pairs, and that is a $\nicefrac{1}{2}$-approximation to \opt. The algorithm (which we denote by $\mathcal{A}$) proceeds as follows: in round $i$, $\mathcal{A}$ selects all the successful edges it has found so far, as well as the edges that \opt\ selects in round $i$ which can augment this matching. We show that in every round the expected size of the matching found by $\mathcal{A}$ is at least a $\nicefrac{1}{2}$-approximation to the expected size of the matching selected by \opt. We provide a complete proof of Proposition \ref{prop:opt-optcommit} in Appendix C.

\smallskip 

In sum, we have considered a number of different matching algorithms. \stablematching\ is decentralized, while the rest are \emph{centralized} matching platforms that determine and prescribe either a myopic maximum matching (\greedycommit) or a non-myopic optimal matching (\opt-Commit, \opt). All of the centralized matching platforms must hide some amount of information from participants, as the prescribed matchings in each round are not necessarily stable and hence not within-round incentive-compatible. In the case of \opt\ the matching proposed by the platform in addition is not across-round compatible, as it needs to coerce successfully matched agents to return and be rematched in future rounds. Both types of incentive constraints are not without loss, and in addition in order to implement either \opt-Commit and \opt\ the platform must solve an NP-hard problem.  Nonetheless all of these platform designs attain an expected reward within a constant-factor approximation of \opt.


\section{Capacitated Matching}\label{sec:4-approx}

 In this section, we consider the difference between stable matching and optimal online matching in more general settings, such as where all agents can be in bilateral matches with multiple of any of the other agents, or agents can be grouped in teams. For example, styling and clothing rental companies may have limited inventory of different items to send to customers, and students may be able to take on multiple projects, or may be grouped into project teams.

Formally, we have a general graph $G(V,E)$ where vertices have capacities $\{C(v)\}_{v \in V}$ given by arbitrary positive integers. In each of $T$ rounds, we can select a capacitated matching among these agents, i.e. some $E' \subseteq E$ such that each $v$ has at most $C(v)$ edges in $E'$ incident to it. As in the previous section, each agent's $v$'s goal is to maximize the expected number of successful matches they participate in over all rounds, and the platform's goal is to maximize the weighted sum of the rewards collected in all $T$ rounds. 

How can we define an appropriate decentralized outcome in this general setting? 
Suppose agents can observe their compatibility $p_{\{i,j\}}$ with every other agent.
A similar argument to the previous section shows that each agent $i$ maximizes their individual expected number of successful matches by matching in each round with the $C(v)$ agents who are most likely to be compatible with $i$ (including agents to whom $i$ has been successfully matched in the past). Hence, in the decentralized matching, there again will be no \emph{blocking pair} of two agents who are not matched to each other but both prefer the other to at least one of their matched partners. Moreover, once a pair of agents is successfully matched they will continue matching with each other in future rounds. This determines a decentralized capacitated stable matching process, which we will again denote by \stablematching. 

\smallskip

In this general setting, we show that \stablematching\ gives a constant-factor approximation to \opt.
\begin{theorem} \label{thm:8-approx} In the setting where agents have arbitrary capacities, the expected reward achieved by the \stablematching\ algorithm is at least a $\nicefrac{1}{11}$ fraction of the expected reward of the optimal online algorithm.
\end{theorem}

Moreover, in many practical matching settings (e.g. gig economy applications, mentorship) matches occur between two sides of a market, and agents on one side of the market (e.g. jobs, mentees) only have capacity 1. In this setting, we can make use of the structure of the underlying graph to refine our approximation factor.

\begin{theorem} \label{thm:7-approx} 
In the many-to-one setting, the expected reward obtained by \stablematching\ is at least a $\nicefrac{1}{7}$ fraction of the expected reward of the optimal online algorithm. 
\end{theorem}

Next we consider a setting where for some set of agents $V$, each subset $S \subseteq V$ of size at most $k$ has an associated probability $p_S$ of being ``compatible" as a team. Generalizing the previous setting, in each of $T$ rounds, the agents are partitioned into teams so that each agent is in at most one team. We can equivalently think of the problem as constructing round-by-round matchings in a random hypergraph. We use modifications of the Domination and Charging Lemmas to show that the decentralized greedy formation of teams of constant size provides a constant-factor approximation to the optimal online algorithm. Formally, we show the following result, and provide a full proof in Appendix E.5.

\begin{theorem}
\label{thm:k_teams}
In a hypergraph where all hyperedges have cardinality at most $k$, the expected reward obtained by \stablematching\ is at least $\nicefrac{1}{2k}$ of the expected reward of \opt. 
\end{theorem}



\subsection{What changes with capacities?}

We first provide some intuition for what differentiates the capacitated setting from the setting in Section~\ref{sec:model}, where agents can only be matched with one other agent every round. This intuition will motivate the new proof techniques used in the capacitated setting, and also highlight the nuances in the decomposition used to prove Theorem~\ref{thrm:2.31}. 

The difficulty in extending the techniques in Section~\ref{sec:model} lies in the definition of what it means to be an augmenting edge, which in turn determines the decomposition of the edges \opt\ selects into $\Aug$ and $\Adj$. A crucial property underlying the proof of the domination lemma is that in the non-capacitated setting, whether a matching $N$ augments another matching $N'$ can be determined by checking whether each $e\in N$ augments $N'$.\footnote{An edge $e$ augments $N'$ if and only if it doesn't share any endpoints with edges in $N'$, a matching $N$ augments $N'$ if and only if it doesn't share any endpoints with edges in $N'$.} Specifically, if each $e \in N$ individually can augment $N'$, and the edges in $N'$ form a matching, then \emph{all} edges in $N'$ can jointly augment $N$. For natural definitions of augmentation in the capacitated setting (e.g. $N$ augments $N'$ if $N \sqcup N'$ is a capacitated matching), a statement of this sort is no longer true. 

Hence, it is natural that in the capacitated case we might give a definition of $\Aug_i$ constructed by jointly considering a group of edges in $\Aug_i$ that can augment $S_{\le i-1}$, instead of considering whether edges in $\Aug_i$ could augment $S_{\le i-1}$ individually. In particular, we might define $\Aug_i$ to be a maximum subset of edges in $O_i$ which can augment $S_{\le i-1}$, and hope that we could claim $\EE[\Aug_i] \le C \cdot \EE[S_{\le i-1}]$ for some constant $C$. This natural approach unfortunately fails; the reason is that $\Aug_i$ can only be formed using information that \stablematching\ did not have available at the start of round $i$.\footnote{In the previous domination lemma, the fact that we could augment edge-by-edge let us argue that $\Aug_i$ could actually be bounded only using information that \stablematching\ had available at round $i$. This subtlety reinforces the need for the careful conditioning in the previous proof.}

Given that these natural generalizations to the augmentation approach fail, we instead decompose successful edges \opt\ selects based on a notion of occupancy for each agent. We say that a node is \emph{heavily occupied} if at least half their capacity is taken up by edges in $S_{\leq t}$. We decompose $O_{\le t}$, the successful edges that \opt\ selects in round $t$, into edges that are in $S_{\le t}$, edges that are not in $S_{\le t}$ but have at least one of their endpoints heavily occupied by $S_{\le t}$, and the remaining edges. Formally, we write 
$$O_{\le t} = (O_{\le t} \cap S_{\le t}) \sqcup \Occ \sqcup \Rem$$
where $\Occ$ denotes all edges in $O_{\le t} \setminus S_{\le t}$ with an endpoint that is filled to at least half-capacity from edges in $S_{\le t}$, and $\Rem$ denotes the remaining edges in $O_{\le t} \setminus (S_{\le t} \sqcup \Occ)$. 

Later in Subsection~\ref{Many-to-One Matchings}, we show a tighter analysis of this same decomposition technique for bipartite many-to-one matchings results in a better approximation factor. In both Sections~\ref{sec:8-approx} and ~\ref{Many-to-One Matchings}, we use a generalized version of the Charging Lemma to bound edges in $\Occ$ incident with nodes that are heavily occupied, and a generalized version of the Domination Lemma to bound edges in $\Rem$ incident only with nodes that are not heavily occupied.

\subsection{A $\nicefrac{1}{11}$-approximation for Capacitated Matchings on General Graphs}\label{sec:8-approx}

In this section, we prove Theorem ~\ref{thm:8-approx}, and show that \stablematching\ achieves at least a $\nicefrac{1}{11}$ fraction of the expected reward of \opt\ in the general capacitated setting. 

Recall that $\RoundAlgSucc{t}$ and $O_{\leq t}$ denote the  \emph{successful} edges selected by \stablematching\ and OPT, respectively, in round $t$.  The proof of Theorem~\ref{thm:8-approx} will follow from proving for all $t \in [T]$
\begin{equation}\label{eq:7-approx}
\mathbb{E}[O_{\leq t}] \leq 11 \cdot \mathbb{E}[S_{\leq t}].
\end{equation}

Recall that we decomposed $O_{\le t} = (O_{\le t} \cap S_{\le t}) \sqcup \Occ \sqcup \Rem$. This decomposition immediately lets us write a generalized Charging Lemma. 

\begin{lemma}[Generalized Charging Lemma] \label{lemma:charging_gen} We have $\EE[\Occ] \le 4 \cdot \EE[S_{\leq t}]$.
\end{lemma}

\begin{proof} Fix any sample graph $G$. For each edge $e \in \Occ$, charge it to an endpoint of $e$ with at least half of its capacity occupied by edges in $S_{\leq t}$. As $\Occ$ forms a valid capacitated matching, the charge placed on each vertex is at most its capacity. Additionally, note that we can cover each of these charges by having every edge in $S_{\le t}$ pay two charges to each of its endpoints. The result follows by averaging over all sample graphs.
\end{proof}

The main claim for which we require different ideas is a generalized Domination Lemma.  Recall that $S_t$ and $O_t$ denote the \emph{new successful} edges selected in round $t$ by \stablematching\ and \opt, respectively. As before, we need to analyze the successful edges selected by \OPT\ based on the round they were discovered; with this in mind we write 
$$\Rem = \Rem_1 \sqcup \Rem_2 \sqcup \ldots \sqcup \Rem_t$$ where $\Rem_i$ denotes those edges in $\Rem$ which OPT selected for the first time in round $i$. The Generalized Domination Lemma bounds $\Rem_i$ against $S_i$. 
\begin{lemma}[Generalized Domination Lemma] \label{lemma:domination_gen} For all $i \le t$ we have $\EE[\Rem_i] \leq 6 \cdot \EE[ S_i ]$.
\end{lemma}

\begin{proof}
Using the same notation in the previous proof of the Domination Lemma, for fixed histories $h$, $h'$ it suffices to show that $$\EE[\Rem_i \mid \mathcal{G}_{h, h'}] \le 6 \cdot \EE[S_i \mid \mathcal{G}_h].$$

Conditioned on our sample graph being in $\mathcal{G}_{h, h'}$, there is a fixed set $N$ of edges that OPT selects in round $i$ which are not in $S_{\le i-1}$, and furthermore have neither endpoint filled to half-capacity by edges in $S_{\le i-1}$. Using the same logic as before, we have that if $N_0 \subseteq N$ denotes those edges in $N$ disjoint from $h$ and $h'$, then $$\EE[\Rem_i \mid \mathcal{G}_{h, h'}] \le \sum_{e \in N_0} p_e$$
as $\Rem_i$ is some (random) subset of $N_0$. In contrast to previous settings, it is not the case that the set of edges in $N_0$ can feasibly augment $S_{\leq i-1}$. However, we show there exists a subset of edges from $N_0$ with at least one-third of the total weight which can feasibly augment $S_{\leq i-1}$. To prove this, we need the following graph theoretic claim.

\begin{claim}
In any weighted graph $G=(V,E)$ where each vertex $v$ has degree $d_v$, there exists a subgraph $S$ whose edges contain at least $\nicefrac{1}{3}$ of the total weight of all edges in $G$, such that each vertex $v$ has degree at most $\lceil \frac{d_v}{2} \rceil$ in $S$. 
\end{claim}

\begin{proof}
Consider the following algorithm for constructing $S$. Of all remaining edges in $E$, add the one of maximum weight to $S$. Call this edge $e = (u,v)$; remove $e$ from $G$. Furthermore, if $u$ has remaining incident edges in $G$, remove one (chosen arbitrarily). If $v$ has remaining incident edges in $G$, remove one (chosen arbitrarily). Continue
until $G$ has no edges remaining. 

Note that for any vertex $v$, there are at most $\lceil \frac{d_v}{2} \rceil$ edges of $S$ incident to $v$. Indeed, after adding any edge to $S$ that is incident to $v$, we delete an edge in $G$ incident to $v$ if possible. Note also that in each step, the weight of the edge we add to $S$ is at least the weight of each edge we delete. Hence in each step, we add to $S$ at least a $\nicefrac{1}{3}$ fraction of the weight of the edge we added and the edges we deleted. This holds in aggregate over all rounds. 
\end{proof}

To finish the proof of the lemma, we use this claim to find a subgraph $S$ of $(V, N_0)$ with at least $\nicefrac{1}{3}$ of the weight; we claim this subgraph is feasible to augment $S_{\leq i-1}$ (and can be selected by \stablematching\ in round $i$). Indeed, note it is disjoint from $h$, and each vertex $v$ is filled to at most $\lceil \frac{C(v)}{2} \rceil$ capacity in $S$ (because $N_0$ is a valid capacitated matching). Furthermore, every endpoint $v$ of an edge in $S$ is filled to at most $\lfloor \frac{C(v)}{2} \rfloor$ capacity by edges in $S_{\le i-1}$. Thus, $S$ is a feasible set of edges to augment $S_{\leq i-1}$ in round $i$. Well-known results also imply that in a fixed round, the \stablematching\ algorithm gives a $\nicefrac{1}{2}$-approximation to maximum weight augmenting edges (see, e.g., \cite{mestre2006greedy}). Hence, $$\EE[S_i \mid \mathcal{G}_h] \ge \sum_{e \in S} p_e \ge \frac{1}{6} \cdot \sum_{e \in N_0} p_e.$$ This demonstrates the result. 
\end{proof}

From  Lemma~\ref{lemma:domination_gen} we can see $$\EE[\Rem] = \sum_{i=1}^t \EE[\Rem_i] \le 6 \sum_{i=1}^t \EE[S_i] \le 6 \cdot \EE[S_{\le t}].$$ Hence, using both Lemmas~\ref{lemma:charging_gen} and Lemma~\ref{lemma:domination_gen} we have $$\EE[O_{\le t}] = \EE[O_{\le t} \cap S_{\le t}] + \EE[\Occ] + \EE[\Rem] \le 11 \cdot \EE[S_{\le t}]$$ 
which completes the proof of Theorem~\ref{thm:8-approx}.

\subsection{A $1/7$-approximation for Many-to-One Matchings}\label{Many-to-One Matchings}

In many practical matching applications, matches occur between two sides of a market, where agents on one side can accept multiple matches while agents on the other side can be matched at most once. For example, workers may take on multiple jobs which each only need one worker, and mentors are frequently matched with multiple mentees who each only have one mentor. In this section, we consider the performance of \stablematching\ in this more restricted setting of many-to-one bipartite matchings. In particular, we show our previous decomposition achieves a better approximation factor.



Formally, we assume our agents are broken into two disjoint sets $S \sqcup M$ (think ``students" and ``mentors"); for each pair $(s, m)$ where $s \in S$ and $m \in M$ we are given a ``success probability" $p_{s,m}$ representing the chance that an edge between $s$ and $m$ is realized when nature samples a random graph. Vertices in $S$ are on the left and vertices in $M$ are on the right. Each agent $m \in M$ has capacity $C(m) \in \mathbb{Z}_{> 0}$, and each agent $s\in S$ has capacity $C(s)=1$.  We show Theorem~\ref{thm:7-approx}, which we restate here.

\setcounter{theorem}{4}
\begin{theorem} \label{thm:7-approx} 
In the many-to-one setting, the expected reward obtained by \stablematching\ is at least a $\nicefrac{1}{7}$ fraction of the expected reward of the optimal online algorithm. 
\end{theorem}

\proof The proof will proceed by showing $\mathbb{E}[O_{\leq t}] \leq 7 \cdot \mathbb{E}[S_{\leq t}]$ for all $t \in [T]$. Partition the edges of $O_{\leq t}$ as
$$O_{\le t} =  \Occ \sqcup \Rem$$ where $\Occ$ denotes edges in $O_{\le t}$ whose left endpoint is incident to some edge in $S_{\le t}$, or whose right endpoint is occupied to at least half capacity by edges in $S_{\le t}$, and $\Rem$ denotes the remaining edges in $O_{\le t}$.\footnote{Note this is an extremely similar decomposition to that in the previous section. We do not have a separate category for edges in $O_{\le t}$ that are also in $S_{\le t}$, as every $e \in O_{\le t} \cap S_{\le t}$ is automatically in $\Occ$.}

\begin{lemma}[Many-to-one Charging Lemma] 
\label{lemma:charging_one_to_many}
We have $ \EE[\Occ] \le 3 \cdot \EE[S_{\le t}].$
\end{lemma}
\begin{proof}
Fix any sample graph. For each edge in $\Occ$, if it is adjacent to an edge in $S_{\le t}$ along its left endpoint, charge it to the unique edge in $S_{\le t}$ it is adjacent to. Otherwise, place a charge on its right endpoint. Each right vertex $v$ that is charged by this process is charged at most $C(v)$ times by this process, and has at least $C(v)/2$ edges in $S_{\le t}$ incident to it. So the charges on right vertices can be covered by charging each edge in $S_{\le t}$ at most twice. In all, we have charged each edge in $S_{\le t}$ at most three times; the result follows by averaging over all sample graphs. 
\end{proof}

\begin{lemma}[Many-to-one Domination Lemma]
\label{lemma:domination_ii_one_to_many}
For all $i \leq t$ we have $\EE[\Rem_i] \le 4 \cdot \EE[S_i].$
\end{lemma}

\proof The intuition behind this lemma is that at round $i$, there exists a feasible subset of the edges in $\Rem_i$ that augments $S_{\leq i-1}$, with at least half the total weight. As before, the greedy property of \stablematching\  algorithm gives a $\nicefrac{1}{2}$-approximation to maximum weight matching on a round-by-round basis. The rigorous proof follows the structure of previous domination lemmas very closely; details are available in Appendix E.4.
\endproof

Combining \Cref{lemma:charging_one_to_many} and \Cref{lemma:domination_ii_one_to_many} gives $$\EE[O_{\le t}] =  \EE[\Occ] + \sum_{i=1}^t \EE[\Rem_i] \le 3 \cdot \EE[S_{\le t}] + \sum_{i=1}^t 4 \cdot \EE[S_i] = 7 \cdot \EE[S_{\le t}],$$ which completes the proof. 
\endproof

\section{Conclusion and Future Directions}\label{sec:conclusion}

This paper contributes to the literature on centralized and decentralized matching in platforms by providing an additional justification for focusing on decentralized algorithms. In particular, we consider matching platforms with repeated interactions between long-lived agents who have unknown but persistent preferences, such as gig economy applications, mentorship matching, kidney exchange, and team formation. We show that letting demand and supply myopically reach a stable matching in a decentralized manner approximates the outcome of computing and imposing a centralized matching. This is despite the fact that the centralized matching processes we consider are not incentive-compatible, and may even be NP-hard to compute.

We focused on a setting for matching with stochastic rewards and learning dynamics. While this setting has been studied in the query-commit literature, the primary motivation in that literature was failure-aware kidney exchange, and the primary focus of that literature was on querying individual edges in sequence to optimize some last-round objective. Beyond kidney exchange, this setting reflects natural learning dynamics that are present in a wide variety of markets that repeatedly match the same agents and provide persistent rewards, such as gig economy job markets, mentorship programs and team formation, and our focus on querying \emph{matchings} and maximizing rewards received \emph{in all rounds} is motivated by these settings. Further work can be done to identify when techniques from one approach can be transferred to the other. 

Another difference from the query-commit literature is that we ask how a decentralized greedy algorithm compares to an optimal online algorithm that is \emph{not restricted to commit} to past successes. While we show that \opt-Commit is a 2-approximation to \opt, further characterization of the relationship between \opt and \opt-Commit remains open. 
In addition, our results hold for any objective that takes a weighted sum of rewards across rounds. This is because the analysis throughout the paper was performed on a round-by-round basis.
Objectives of potential interest included the sum of rewards across different rounds, expected discounted rewards, as well as the size of maximum successful matching that can be identified at the end of $T$ rounds, and we believe our analysis can be tightened for these specific objectives to produce sharper bounds. We also believe that there are sharper bounds for the capacitated settings. We leave the study of tighter approximation guarantees to future work.
    
    
         
More broadly, our findings raise a number of follow-up questions. Can platforms be designed in a way to nudge agents towards a decentralized outcome that approximates the optimal centralized matching achievable by the platform? What is the approximation gap between incentive-compatible matching platforms and what is achievable by a centralized authority? This paper also explores the problem of learning through prior assignments, a natural direction that is relatively less studied in the literature. Variations on this theme merit future exploration; for example, matching programs in highly relational and idiosyncratic settings like mentorship matching frequently face a \emph{cold start} problem, where the program designer has some prior over which features best predict a successful match, and can refine this prior by observing the outcomes of prior matches. 
In mentorship or rotation programs an additional feature is that edge successes are often \emph{correlated}; for example a student may think she has a general interest in economics and computation but learn that she is more interested in theory generally. 
We hope that this paper will motivate future work in these directions.

\bibliographystyle{plain}
\bibliography{references}

\newpage
\appendix

\section{Optimum Offline vs Optimum Online}\label{appendix:optoffline}

To see that approximation to the optimal offline algorithm can be arbitrarily bad, consider a bipartite graph $K_{n,n}$ where each edge has probability $\nicefrac{1}{n}$ of being successful, and where we only have one round. Clearly no online algorithm gets expectation better than $1$. However, the optimal offline algorithm can simply see the realization of each edge, and then select the maximum matching; we claim that the expected size of this maximum matching is has size $\Theta(n)$. 

We give a loose bound; observe that for any edge $e$, the probability it is included in the maximum matching of the realized graph is at least the probability that $e$ is realized and no edges adjacent to it are realized. In particular, this probability is at least $\frac{1}{n} \cdot \left( 1 - \frac{1}{n} \right)^{2n-2} \ge \frac{1}{n} \cdot \frac{1}{e^2}$. By linearity of expectation, the expected size of the maximum matching in the realized graph is hence at least $\frac{n}{e^2}$. 

\section{The \greedycommit\ Algorithm}\label{app:greedycommit}
Here we provide a formal description of the \greedycommit\ algorithm.

 \begin{tcolorbox}
 
 \textbf{The Greedy-Commit Algorithm}
 
Initialize the set of successful edges $A \leftarrow \emptyset$. Let $S \leftarrow V.$ 

For rounds $1 \leq t \leq T$:
\begin{itemize} 
    \item Find a matching $M_t$ between vertices in $S$ maximizing $\sum_{e\in M_t} p_e.$ 
    \item Output $M_t \cup A$  as the selected matching for round $t$.
    \item If an edge $e$ in $M_t$ is successful, then add it to $A$ and remove both its endpoints from $S$.  Otherwise, set $p_e \leftarrow 0$.
\end{itemize}

\end{tcolorbox}

\section{\OPT-Commit vs \OPT}\label{appendix:greedy-vs-greedy-commit}


Recall that \OPT-Commit is the optimal matching algorithm subject to the restriction that any edges it matches that are successful must also be matched in all future rounds. We first demonstrate that the \OPT-Commit algorithm is not the same as \opt. 

\begin{proposition}
\OPT-Commit algorithm is not the same as \opt.
\end{proposition}

\begin{proof}
Consider a complete bipartite graph with bipartition $(\{u_1,u_2\},\{u_3,u_4\})$ such that each edge is realized with probability $0.7$; say we are matching over two rounds. We show that \OPT-Commit and \OPT\ perform differently on this problem instance.

In the first round, \OPT\ and \OPT-Commit both select a matching of size 2. Without loss of generality we assume they both select edges $(u_1,u_3)$ and $(u_2,u_4)$ in the first round. If both $(u_1,u_3)$ and $(u_2,u_4)$ are realized, or if neither $(u_1,u_3)$ nor $(u_2,u_4)$ is realized, both algorithms will behave identically in the second round.

However, if only one edge is successful (say, $(u_1,u_3)$), \OPT-Commit selects only edge $(u_1,u_3)$ in the second round, receiving reward 1 for the second round, while \opt\ selects edges $(u_1,u_4)$ and $(u_2,u_3)$ in the second round, hence receiving expected reward $2 \cdot 0.7 = 1.4$ in this round.
 Hence \opt\ and \opt-Commit behave differently, and the expected reward achieved by \opt\ can be higher than that achieved by \opt-Commit. Note that the exact same example additionally shows that the \textsc{Greedy} algorithm, which queries the matching with the highest expected reward in each round, is different from the \greedycommit\ algorithm.
 \end{proof}
 
Next we prove Proposition~\ref{prop:opt-optcommit} from Section~\ref{sec:upperbound}, which shows that \OPT-Commit is a $\nicefrac{1}{2}$-approximation to \OPT. We restate the proposition here.
    
\setcounter{proposition}{1}
\begin{proposition}
The expected reward of \opt-commit is at least a $\nicefrac{1}{2}$-approximation to \opt.
\end{proposition}
\setcounter{proposition}{3}
\begin{proof}
We consider the algorithm $\mathcal{A}$, which in round $i$, selects all the successful edges it has found so far, as well as the edges that \opt\ selects in round $i$ which can augment this matching. For any set of edges $E$, we let $N(E)$ denote the set of all vertices incident to at least one edge in $E$. 

Let $S_{\leq t}$ and $O_{\leq t}$ denote the successful edges selected by $\mathcal{A}$ and OPT respectively in round $t$. The proof proceeds by showing $\EE[O_{\leq t}] \leq 2 \cdot \EE[S_{\leq t}]$ for all $t \in [T]$; the result will follow as clearly OPT-commit performs at least as well as $\mathcal{A}$.

Let $S_i^+$ and $O_i^{+}$ denote the \textit{new} edges selected by $\mathcal{A}$ and OPT respectively in round $i$, and let $S_i$ and $O_i$ denote the \textit{successful} edges in $S_i^+$ and $O_i^{+}$ respectively. Fix any sample graph $G$. For each round $i$, observe that the new edges $\mathcal{A}$ selects in round $i$ are precisely 
$$S_i^{+} := O_i^+ \setminus \left \{ e \in O_i^+ : N( \{ e \}) \cap N(S_1) \sqcup N(S_2) \sqcup \cdots \sqcup N(S_{i-1}) \neq \emptyset \right \}.$$ Note that each successful edge OPT selects in round $t$ is either selected by $\mathcal{A}$ as well, in which case $e \in S_{\leq t}$, or edge $e$ is adjacent to a node in $N(S_1) \sqcup N(S_2) \sqcup \cdots \sqcup N(S_{t})$. (If edge $e$ was selected for the first time by OPT in round $j$, and if edge $e$ was not adjacent to any successful edge selected by $\mathcal{A}$ up to but not including round $j$, $e$ would've been selected in round $j$ by $\mathcal{A}$ as well.) As the edges of $O_{\leq t}$ form a valid matching, for any fixed sample graph $|O_{\leq t}| \leq 2 \cdot |S_{\leq t}|$. Averaging over all sample graphs, $\EE[O_{\le t}] \le 2 \cdot \EE[S_{\le t}]$. 
\end{proof}


\section{Computing \opt\ is NP-hard} \label{appendix:optnphard}

In this section, for completeness, we review the reduction \cite{chen2009approximating} use to show that computing OPT-commit is NP-hard. We briefly note that the same ideas show that computing OPT is NP-hard. 

Chen et al. reduce from the problem of determining whether a graph $G = (V, E)$ is $k$-edge-colorable. Given a graph $G=(V,E)$ with $m = |E|$, they construct a stochastic matching instance on $G$ over $k$ rounds where each each edge has probability $\frac{1}{m^3}$ of being realized; our objective is the total number of successful edges queried on the $k$\textsuperscript{th} round. If $G$ is $k$-edge-colorable, a feasible strategy is to commit to all the successful edges we have seen thus far, and augment with all edges of color $i$ that we can in the $i$\textsuperscript{th} round. The expected reward of \opt\ in this case is at least $$ \sum_{e \in E} \frac{1}{m^3} - \binom{m}{2} \left( \frac{1}{m^3} \right)^2 \ge \frac{1}{m^2} - \frac{1}{m^4}.$$ Indeed, note that $\sum_{e \in E} \frac{1}{m^3}$ gives the expected number of total successful edges found if we simply queried all edges of color $i$ in round $i$. The expected number of edges this overcounts compared to our actual strategy is at most the expected number of pairs of edges that are adjacent, which is at most $\binom{m}{2} \left( \frac{1}{m^3} \right)^2$. 

If $G$ is not $k$-edge-colorable, in $k$ rounds we know at most $m-1$ total edges can be queried. Hence the expected reward of OPT would certainly be at most $$ (m-1) \cdot \frac{1}{m^3} = \frac{1}{m^2} - \frac{1}{m^3}$$ as this upper bounds the number of distinct successful edges \opt\ could find across all rounds. Hence, computing the expected reward of OPT suffices to determine whether $G$ is $k$-edge-colorable. 

\section{Additional Omitted Proofs}

\subsection{Proof of Lemma \ref{lemma:domination} (Refined Domination Lemma)} \label{refineddomlemmaproof}
\begin{proof}
The edges in $$\Aug_i \sqcup \Adj_i(S_j) \sqcup \Adj_i(S_{j+1}) \sqcup \ldots \sqcup \Adj_i(S_t)$$ are all selected by OPT for the first time in round $i$, and form a matching. Furthermore, they can augment $S_{\le j-1}$. For convenience, denote these edges by $M$. 

We only require small changes from the proof of the Domination Lemma in the previous section. In particular, let $h$ be any possible history on the first $j-1$ rounds, and let $h'$ be any possible history on the first $i-1$ rounds. If $\mathcal{G}_h$ denotes the set of all base graphs $G$ which induce a history of $h$ on \stablematching\ (for $j-1$ rounds), and $\mathcal{G}_{h,h'}$ denotes the subset of those that additionally induced a history of $h'$ on OPT (for $i-1$ rounds), it suffices to show that $$2 \cdot \EE[S_j \mid \mathcal{G}_h] \ge \EE [ M \mid \mathcal{G}_{h, h'}].$$ 

Conditioned on our sample graph being in $\mathcal{G}_{h, h'}$, there is a fixed set $N$ of edges that OPT selects in round $i$ which can augment $S_{\le j-1}$. Using the same logic as before, we have that if $N_0 \subseteq N$ denotes those edges in $N$ disjoint from $h$ and $h'$, then $$\EE[M \mid \mathcal{G}_{h, h'}] \le \sum_{e \in N_0} p_e,$$ as $M$ is clearly a subset of the edges that OPT selects in round $i$ which can augment $S_{\le j-1}$. We also note that \stablematching\ in round $j$ selects edges with at least half the total weight of those in $N_0$ (they are all disjoint from $h$), so $$2 \cdot \EE[S_j \mid \mathcal{G}_h] \ge \sum_{e \in N_0} p_e.$$ This demonstrates the result. \end{proof}

\subsection{Proof of \Cref{thrm:greedy-commit} (lower bound for \greedycommit)}\label{greedycommitanalysis}

Here we prove that \greedycommit\ is at least a $0.43$-approximation to \opt. As the proof is very similar to the proof that \stablematching\ is a $0.316$-approximation, we only mention the key places where the proof differs. When analyzing \greedycommit\, we let $S_{\le t}$ denote the successful edges that \greedycommit\ selected in round $t$, and let $S_t$ denote the new successful edges that \greedycommit\ selected in round $t$. All other sets of edges are defined verbatim as before. 

The Charging Lemma, which stated that $\sum_{i=1}^t \EE[\Adj_i(S_j)] \le 2 \cdot \EE[S_j]$, holds with the same proof. However, we prove a stronger version of the Refined Domination Lemma. In particular, we show that for any $1 \le i, j \le t$ we have $$\EE[\Aug_i] + \EE[\Adj_i(S_j)] + \EE[\Adj_i(S_{j+1})] + \ldots + \EE[\Adj_i(S_t)] \le \EE[S_j].$$ The proof in \Cref{refineddomlemmaproof} holds verbatim, with the exception of penultimate sentence. Instead, we claim that \greedycommit\ in round $j$ selects edges with at least the total weight of those in $N_0$ (they are all disjoint from $h$, and hence \greedycommit\ selects something at least as good as $N_0$), so $$\EE[S_i \mid \mathcal{G}_h] \ge \sum_{e \in N_0} p_e.$$
This proves a strengthened domination lemma. With these lemmas in place, we get the following factor-revealing LP, following the techniques we used to analyze \stablematching\ against \opt. 

\vspace{0.3em}

\textbf{Factor-revealing LP (Primal)}
\setcounter{equation}{0}
\begin{align}
\nonumber \max\;\;\;\;\;\; \sum_{i = 1}^{t} X_{i} +  \sum_{i = 1}^{t}  \sum_{j = 1}^{t} & X_{i,j} &&\text{ }  \\
\text{such that \;\;\;\;\;\;\;}
X_{i} + \sum_{q = j}^{t} X_{i,q} &\leq Y_j \;\; &&\text{for all } 1 \le i, j \le t  \\
\sum_{i = 1}^{t} X_{i,j} &\leq 2 Y_j \;\; &&\text{for all } 1 \le j \le t \\
\sum_{j = 1}^{t} Y_{j} &\leq 1 \;\; &&\text{  }\\
 Y_j \geq 0, \; X_{i} \geq 0, \; & X_{i,j} \geq 0  \;\; &&\text{for all }  1 \le i,j \le t 
\end{align}

To give a bound on the maximum value obtained by this LP, we take the dual, noting that by weak duality it suffices to analyze the value obtained by a specific feasible solution. The dual of our LP is given below. 

\textbf{Factor-revealing LP (Dual)}
\setcounter{equation}{0}
\begin{align}
\nonumber \min\;\;\;\;\;\;  u \hfill  && \text{ } \\
\text{such that \;\;\;\;\;\;\;}
\sum_{q = 1}^{j} F_{i,q} + c_j &\geq 1 \;\; &&\text{for all } 1 \le i, j \le t \\
\sum_{i = 1}^{t} F_{i,j} + 2c_j &\leq u \;\; &&\text{for all } 1 \le j \le t \\
\sum_{j = 1}^{t} F_{i,j} &\geq 1 \;\; &&\text{for all } 1 \le i \le t \\
 u \geq 0, \; c_{j} \geq 0, \; & F_{i,j} \geq 0  \;\; &&\text{for all } 1 \le i, j \le t
\end{align}

Setting the dual variables as
$$F_{i,j} = \frac{2}{t(e^2 - 1)} \cdot \left( \frac{t}{t-2} \right)^{j-1} \text{\;\;\; for } 1 \le i, j \le t $$

$$ c_j =   1 - \frac{t^j - (t-2)^j}{t(t-2)^{j-1}(e^2 - 1)} 
\text{\;\;\; for } 1 \le j \le t $$
$$ u= \max_{j \in [t] } \left \{  \sum_i F_{i,j}+2c_j  \right \} = 2 + \frac{2(t-2)^t}{2(t-2)^{t-1} + t^t - t(t-2)^{t-1}},$$ results in a feasible dual solution, where the objective value $u(t)$ is increasing in $t$. Moreover $$\lim_{t\rightarrow\infty} \left( 2 + \frac{2(t-2)^t}{2(t-2)^{t-1} + t^t - t(t-2)^{t-1}} \right) = 2 + \frac{2}{\mathrm{e}^2-1}.$$

Hence by weak duality, the objective value of our primal is at most $2 + \frac{2}{e^2-1}$. It follows that the expected reward of Greedy-Commit in round $t$ is at least a $ \left( 2 + \frac{2}{e^2-1} \right)^{-1} \ge 0.43$ fraction of the expected reward of OPT in round $t$.

\subsection{Proof of Lemma~\ref{lem:UB} (upper bound for \stablematching\ and \greedycommit)}
 Here we illustrate the graph $G_{n\epsilon}$ used to  show that \stablematching\ and \greedycommit\ are at best a $\nicefrac{1}{2}$ approximation for \opt.
 \begin{figure}[H]
  \includegraphics[width=6cm]{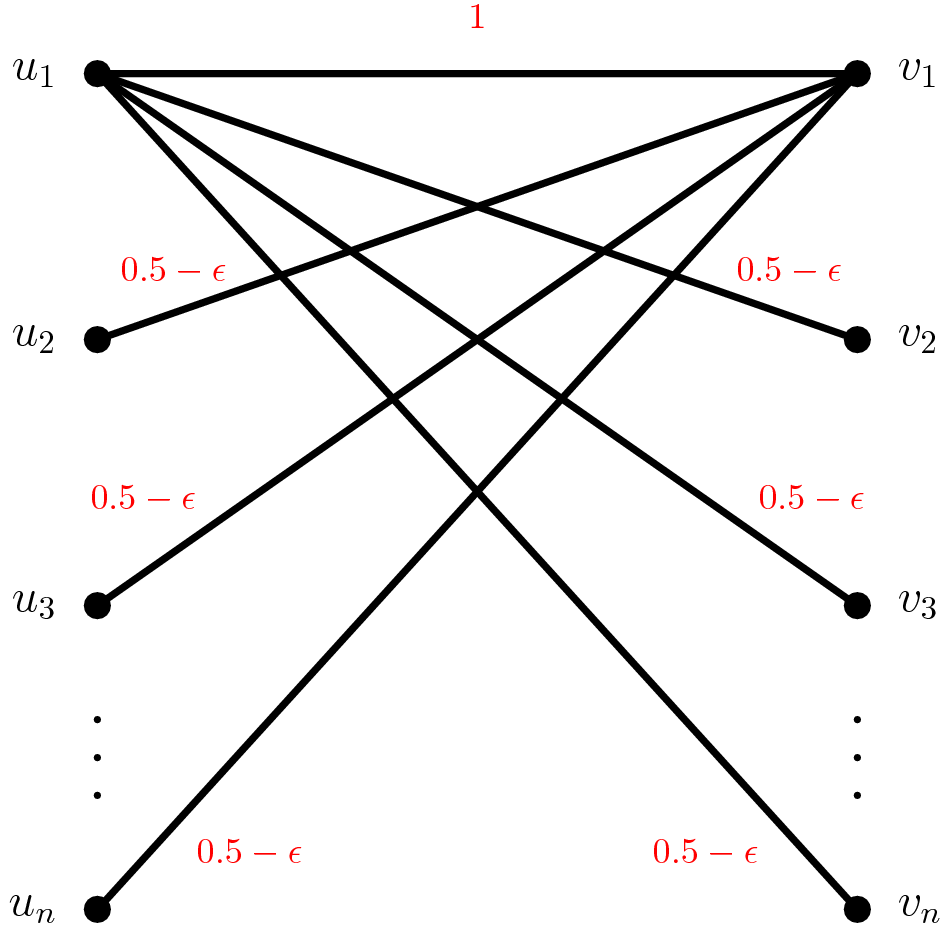}
  \caption{Graph $G_{n, \epsilon}$ with edge realization probabilities}
  \label{fig:Greedy2approx}
\end{figure}

\subsection{Proof of Lemma \ref{lemma:domination_ii_one_to_many} (Many-to-one Domination Lemma)}\label{appendix:domination_ii_one_to_many}

\begin{proof}
Let $h$ and $h'$ denote any two fixed histories on $i-1$ rounds. It suffices to show that 
$$\EE[\Rem_i \mid \mathcal{G}_{h, h'}] \le 4 \cdot \EE[S_i \mid \mathcal{G}_{h, \cdot}].$$ Conditioned on our sample graph being in $\mathcal{G}_{h, h'}$, there is a fixed set $N$ of edges that OPT selects in round $i$ which are not in $S_{\le i-1}$, and furthermore have their left endpoints unoccupied by $S_{\le i-1}$, and their right endpoints only occupied to at most half capacity by $S_{\le i-1}$. If $N_0 \subseteq N$ denotes those edges in $N$ disjoint from $h$ and $h'$, then 
$$\EE[\Rem_i \mid \mathcal{G}_{h, h'}] \le \sum_{e \in N_0} p_e,$$
as $\Rem_i$ is some (random) subset of $N_0$. In contrast to the previous setting, it is not the case that all edges in $N_0$ can augment the matching $S_{\leq i-1}$. However, there exists a subset of edges from $N_0$ with at least one-half of the weight, in which can feasibly augment the matching $S_{\leq i-1}$. Indeed, $N_0$ forms a valid capacitated matching. Every edge in $N_0$ also has its left endpoint not occupied by $S_{\le i-1}$ and its right endpoint only occupied to at most half capacity by $S_{\le i-1}$. Hence, greedily augmenting $S_{\le i-1}$ with edges from $N_0$ vertex-by-vertex captures at least half the weight of $N_0$. We also know that in each round, the \stablematching\ algorithm gives a $\frac{1}{2}$-approximation to maximum weight matching we could augment, even in the capacitated case (see \cite{mestre2006greedy}). Hence $$\EE[S_{i} \mid \mathcal{G}_{h, \cdot}] \ge \frac{1}{4} \cdot \sum_{e \in N_0} p_e$$
which completes the proof.
\end{proof}

\subsection{Proof of Theorem~\ref{thm:k_teams}: A $\nicefrac{1}{2k}$-approximation for Matching Teams of Size $k$}\label{appendix:k_teams}

Consider the matching on hypergraphs problem where each round we match teams of size at most $k$. Our main result is Theorem~\ref{thm:k_teams}, which we restate here.

\begin{theorem}
In a hypergraph where all hyperedges have cardinality at most $k$, the expected reward obtained by \stablematching\ is at least $\nicefrac{1}{2k}$ of the expected reward of \opt. 
\end{theorem}

\begin{proof}[Proof Sketch] As the proof follows along similar lines to previous proofs, we provide only a sketch. The proof proceeds by showing $\mathbb{E}[O_{\leq t}] \leq 2k \cdot \mathbb{E}[S_{\leq t}]$ for all $t \in [T]$, using modifications of the Charging Lemma and Domination Lemma. We use the same decompositions as in \Cref{sec:model} to define $\{\Aug_i\}$ and $\{\Adj_i(S_j)\}$ in the hypergraph setting. 
\end{proof}

\begin{lemma}[Hypergraph Charging Lemma] \label{lemma:hypergraph_charging}
For all $j \le t$, 
$\sum_{i=1}^t \EE[\Adj_i(S_j)] \le k \cdot \EE[S_j].$
\end{lemma}

\begin{proof}[Proof Sketch]
The proof follows from the same argument as in the proof of Lemma \ref{lemma:charging2}. In particular, as hyperedges have cardinality at most $k$, each hyperedge can be charged at most $k$ times. 
\end{proof}   

\begin{lemma}[Hypergraph Domination Lemma]\label{lemma:hypergraph_domination}  Then, for all $i \leq t$, $\EE[\Aug_i] \le k \cdot \EE[S_i]$.
\end{lemma}

\begin{proof}[Proof Sketch]
The proof of \Cref{lemma:dominationlemma} holds almost verbatim. However, where we previously argued that the greedy algorithm gives a $\nicefrac{1}{2}$-approximation for max-weight matching, we now argue that it gives a $\nicefrac{1}{k}$-approximation for max-weight matching in a hypergraph where all hyperedges have cardinality at most $k$, (see, e.g., \cite{mestre2006greedy}).
\end{proof}

Now, using the two previous lemmas, we can bound 
\begin{align*}
\EE[O_{\le t}] &= \sum_{i=1}^t \Big( \EE[\Aug_i] + \EE[\Adj_i] \Big) 
\le \sum_{i=1}^t k\EE[S_i] + \sum_{i=1}^t \EE[\textsc{Adj}_i] 
\leq 2k \cdot \EE[S_{\le t}],
\end{align*}
which completes the proof.
\endproof

\end{document}